\documentclass{ifacconf}

\usepackage{graphics} 
\usepackage{epsfig} 
\usepackage{amsmath} 
\usepackage{amssymb}  
\let\theoremstyle\relax  
\usepackage{amsthm}

\usepackage{algpseudocode}
\usepackage{algorithm}
\usepackage{epstopdf}

\usepackage{verbatim}
\usepackage{color}
\usepackage{enumitem}

\usepackage{selinput}
\usepackage{tikz}
\usepackage{pstricks}
\usetikzlibrary{calc, shapes, backgrounds, bending, decorations.pathmorphing, positioning, decorations.pathreplacing}
\usepgflibrary{arrows}
\usepackage{pgfplots}
\usetikzlibrary{decorations.markings}
\makeatletter
\tikzset{
	nomorepostactions/.code={\let\tikz@postactions=\pgfutil@empty},
	mymark/.style 2 args={decoration={markings,
			mark= between positions 0 and 1 step (1/6)*\pgfdecoratedpathlength with{%
				\tikzset{#2,every mark}\tikz@options
				\pgfuseplotmark{#1}%
			},  
		},
		postaction={decorate},
		/pgfplots/legend image post style={
			mark=#1,mark options={#2},every path/.append style={nomorepostactions}
		},
	},
}
\makeatother 
\pgfplotsset{compat=newest} 
\pgfplotsset{plot coordinates/math parser=false} 
\pgfplotsset{every tick label/.append style={font=\small}}
\pgfplotsset{every label/.append style={font=\small}}
\pgfplotsset{every x tick scale label/.style={
		at={(1,0)},xshift=1pt,anchor=south west,inner sep=0pt}}
\newlength\figureheight 
\newlength\figurewidth
\usetikzlibrary{decorations.pathreplacing}

\usepackage{natbib}        

\usepackage{setspace}

\newcommand{\norm}[1]{\left\lVert#1\right\rVert}
\def\bl{\bar{\lambda}}
\def\Hnorm{\upharpoonleft \hspace{-0.13cm} H \hspace{-0.1cm} \upharpoonright}
\newcommand{\HnormB}[1]{\left.\upharpoonleft \hspace{-0.15cm} #1 \hspace{-0.15cm} \right.\upharpoonright} 

\newtheorem{theo}{Theorem}
\newtheorem{lemma}{Lemma}
\newtheorem{propo}{Proposition}

\theoremstyle{definition}
\newtheorem{defi}{Definition}
\newtheorem{ass}{Assumption}
\newtheorem{remark}{Remark}

\usepackage{units}
\usepackage{url}

\begin{document}

\begin{frontmatter}
\title{Stability and performance in transient average constrained economic MPC without terminal constraints\thanksref{footnoteinfo}} 

\thanks[footnoteinfo]{The authors thank the German Research Foundation (DFG) for support of this work under Grants GRK 2198/1 - 277536708 and AL 316/12-2, and MU 3929/1-2 - 279734922.}

\author[First]{Mario Rosenfelder} 
\author[Second]{Johannes K\"ohler}
\author[Second]{Frank Allg\"ower}

\address[First]{M.Sc. student University of Stuttgart, 70550 Stuttgart, Germany (email:mario-rosenfelder@web.de).}
\address[Second]{Institute for Systems Theory and Automatic Control, University of Stuttgart, 70550 Stuttgart, Germany (email:$\{$johannes.koehler, frank.allgower\}@ist.uni-stuttgart.de).}
\begin{abstract}
	In this paper, we investigate system theoretic properties of transient average constrained economic model predictive control (MPC) without terminal constraints. We show that the optimal open-loop solution passes by the optimal steady-state for consecutive time instants. Using this turnpike property and suitable controllability conditions, we provide closed-loop performance bounds. Furthermore, stability is proved by combining the rotated value function with an input-to-state (ISS) Lyapunov function of an extended state related to the transient average constraints. The results are illustrated with a numerical example.
\end{abstract}
\end{frontmatter} 
{\let\thefootnote\relax\footnote{{\copyright 2020 the
authors. This work has been accepted to IFAC for publication under a Creative Commons Licence CC-BY-NC-ND}}}

\section{Introduction}
\subsubsection*{Motivation:}
Model predictive control (MPC)~\citep{rawlings2017model} is a popular control method that computes the control input by repeatedly solving an optimal control problem. The prime advantages of MPC are that it can deal with complex nonlinear dynamics, general objective functions,  multiple-input-multiple-output (MIMO) systems, as well as arbitrary input and state constraints. 

The main objective of MPC does not necessarily need to be stability but can rather be optimal performance with respect to an economic criterion, which results in a cost function that does not have to be positive definite with respect to any setpoint. This variant of MPC is called economic MPC (EMPC) \citep{Angeli12}, \citep{Faulwasser18}. Usually, an optimal steady-state is determined and then, is used as a terminal condition for the finite horizon problem. These terminal conditions are often omitted in practical applications since they can be complicated to design (additional offline computation) and they can limit the operating region of the controller. Moreover, the absence of additional constraints makes the finite horizon optimal control problem in each step easier to solve. 

Additionally to point-wise in time constraints, it stands to reason to consider constraints on average values. Constraints on states and inputs averaged over some finite time period can be of interest in several applications. For example, overheating of electric motors can be avoided by limiting values over a period of time. Another example are chemical processes where the amount of inflow must not exceed a certain value over a finite time period or, limiting the frequency deviation in power grids. Hence, the question arises whether we can consider such transient average constraints without imposing terminal constraints in order to benefit from the advantages of both properties. We derive theoretical guarantees in terms of performance guarantees as well as stability for the transient average constrained EMPC scheme without terminal constraints.
\subsubsection*{Related work:}
Economic MPC has been investigated with a terminal equality constraint~\citep{Diehl2011}, as well as for a terminal cost and a terminal region~\citep{Amrit11}. Performance estimates can be found in~\citep{Angeli12} and~\citep{GruenePanin2015}.
EMPC without terminal constraints is introduced by~\cite{Gruene13} where the results are further developed by~\cite{Gruene14} in order to show practical asymptotic stability. In~\cite{Mueller13,Mueller14b}, convergence of averagely constrained EMPC with terminal ingredients is considered. \cite{Koehler17} present a transient, nonaveraged performance estimate for the corresponding closed loop with asymptotic average constraints. The stricter form of transient average constraints is introduced by~\cite{Mueller14}. There, closed-loop average performance bounds and convergence results are proved for EMPC with transient average constraints by imposing a terminal region and a terminal cost.
\subsubsection*{Contribution:} 
So far, results for transient average constraints in EMPC have been shown by imposing terminal conditions~\citep{Mueller14}. However, transient average constrained EMPC without terminal constraints has not been investigated. We bridge this gap by the following contributions.

We describe the EMPC scheme with transient average constraints using an extended state containing past auxiliary outputs.
As a first contribution, we extend existing turnpike arguments to conclude a turnpike property over multiple consecutive time steps, which implies a turnpike for this extended state. 
Then we provide transient performance guarantees and show value convergence of the closed-loop cost functional, similar to the derivation of \cite{Gruene13} and \cite{Gruene14}.
For the stability analysis, we show that contrary to most EMCP schemes, in the considered formulation the rotated value function is \textit{not} a suitable Lyapunov function.
Instead, we use a Lyapunov function consisting of the rotated value function and an input-to-state (ISS) Lyapunov function that describes the finite-memory property of added state variables.
With this novel Lyapunov function, we prove practical asymptotic stability of the closed loop.
We illustrate the results with the academic example from~\citep{Mueller13,Koehler17}. 

\subsubsection*{Outline:} Section~\ref{sec:Preliminiaries} formulates the control problem.  Section~\ref{sec:Turnpikes} provides turnpike properties and Section~\ref{sec:loc_cont} shows local continuity of the value function. Section~\ref{sec:peformance} contains performance guarantees. 
Section~\ref{sec:Stability} derives practical asymptotic stability of the closed loop. Section~\ref{sec:Example} illustrates the results with a numerical example.
Section~\ref{sec:Conclusion} concludes the paper. We note that the results in this paper are based on the thesis~\citep{thesis_mario}, which is available online and contains more detailed proofs.

\subsubsection*{Notation:}
The set of integers in $[a,b]\subseteq\mathbb{R}$ is denoted by $\mathbb{I}_{[a,b]}$, and the set of integers greater or equal to $a$ is denoted by $\mathbb{I}_{\geq a}$. 
We denote a ball with radius $r$ around a point $y$ by $\mathcal{B}_r (y):=\{ x\in\mathbb{R}^n \vert \norm{x-y}\leq r \}$.
For $c\in\mathbb{R}$, $\lceil c \rceil$ is defined as the smallest integer greater or equal to $c$.
With $\mathcal{K}$ we denote the set of continuous, strictly increasing functions $\alpha:\, [0,a)\to[0,\infty)$, which satisfy $\alpha(0)=0$. For $a=\infty$ and $\alpha(r)\to \infty$ as $r\to\infty$ we denote the class $\mathcal{K}_\infty$. The set of all decreasing functions $\delta:\, \mathbb{N}_0\to[0,\infty)$ with $\lim_{k\to\infty} \delta(k)=0$ is denoted by $\mathcal{L}_{\mathbb{N}}$. Class $\mathcal{KL}$ is the set of all continuous functions $\beta:\, [0,a)\times[0,\infty)\to[0,\infty)$ for which it holds $\beta(\cdot,s)\in\mathcal{K}$ and $\beta(r,\cdot)\in\mathcal{L}_{\mathbb{N}} $. Furthermore, we denote by $\mathcal{KLS}$ the class of functions $\beta\in\mathcal{KL}$ for which $\sum_{k=0}^\infty \beta(r,k)$ is finite for all $r\geq0$ and for which $\gamma_\beta (\cdot):=\sum_{k=0}^\infty \beta(\cdot,k)\in\mathcal{K}$.
\section{Preliminaries and problem setup}\label{sec:Preliminiaries}
\subsubsection*{Problem Setup:}
We consider discrete-time nonlinear systems 
\begin{equation}\label{eq:system}
x(k+1)=f(x(k),u(k)),
\end{equation}
with a continuous map $f:\, \mathbb{R}^n\times \mathbb{R}^m\to\mathbb{R}^n$, state $x\in\mathbb{R}^n$ and control values $u\in\mathbb{R}^m$. The system is subject to state and input constraints, which can possibly be coupled, i.\,e., $(x(k),u(k))\in\mathbb{Z}$ with a compact set $\mathbb{Z}\subseteq\mathbb{X}\times \mathbb{U}$ where $\mathbb{X}\subseteq\mathbb{R}^n$ and $\mathbb{U}\subseteq\mathbb{R}^m$. Additionally, the system is subject to average constraints expressed in terms of an auxiliary output $y=h(x,u)\in\mathbb{R}^p$. Considering transient average constraints, we require that for some given time period $T\in\mathbb{I}_{\geq 1}$ for all $k\geq 0$ it holds
\begin{equation}
\label{eq:trans_av_con}
\sum_{j=k}^{k+T-1}\frac{h(x(j),u(j))}{T}\in\mathbb{Y}.
\end{equation}
In the following, we consider w.\,l.\,o.\,g. $\mathbb{Y}:=\mathbb{R}^p_{\leq 0}$. In case $T=1$, we recover the special case of standard point-wise in time constraints.

For a given control sequence $u_N \in\mathbb{U}^N$ we denote the solution of~\eqref{eq:system} by $x_{u_N}(k,x)$ where $x\in\mathbb{X}$ is the initial value.
Furthermore, system~\eqref{eq:system} has a continuous stage cost $\ell:\, \mathbb{Z}\to\mathbb{R}$ which is assumed to be bounded from below. The standing assumptions are summarized as follows:
\begin{ass}\label{ass:comp_cont}
	The constraint set $\mathbb{Z}$ is compact and the maps $f:\mathbb{Z}\to\mathbb{X}$ and $\ell:\mathbb{Z}\to\mathbb{R}$ are continuous, i.\,e., there exist $\alpha_f,\, \alpha_l\in\mathcal{K}_\infty$ such that it holds $\norm{f(x_1,u_1)-f(x_2,u_2)}\leq \alpha_f(\norm{(x_1-x_2,u_1-u_2)})$ and $\left| \ell(x_1,u_1)-\ell(x_2,u_2) \right|\leq \alpha_l(\norm{(x_1-x_2,u_1-u_2)})$ for all $(x_1,u_1),(x_2,u_2)\in\mathbb{Z}$. Furthermore, the map $h:\,\mathbb{Z}\to\mathbb{R}^p$ is Lipschitz continuous with constant $L_h>0$.
\end{ass} 
The control goal is to minimize the stage cost $\ell(x,u)$ over the prediction horizon $N$ for system~\eqref{eq:system} subject to point-wise in time constraints and transient average constraints~\eqref{eq:trans_av_con}. Given an initial state $x$, the open-loop costs of a control sequence $u(\cdot)\in\mathbb{U}^N$ are defined as
\begin{equation}\label{eq:def:cost_functional}
J_N (x,u):=\sum_{k=0}^{N-1} \ell(x_u(k,x),u(k)).
\end{equation} 
Since feasibility of input sequences for transient average constrained EMPC also depends on past auxiliary output values, we introduce the additional state
$H(k):= [h(x(k-T+1),u(k-T+1)), \ \dots, \ h(x(k-1),u(k-1))  ]\in\mathbb{R}^{p\times (T-1)}$
and write $H_j$ for the $j$-th column of $H$. Analogous to the point-wise feasible set $\mathbb{Z}$, we write  $\mathbb{H}:=\{H\in\mathbb{R}^{p\times (T-1)}|~\underline{h}\leq H_{j} \leq \overline{h},~j\in\mathbb{I}_{[1,T-1]}\}$ with $\overline{h}_i:=\sup_{(x,u)\in\mathbb{Z}}h_i(x,u)$, $\underline{h}_i:=\inf_{(x,u)\in\mathbb{Z}}h_i(x,z)$
for the set of all feasible $H$. Now, given a state $(x,H)\in\mathbb{X}\times\mathbb{H}$, the set of all admissible control sequences is denoted by $\mathbb{U}^N(x,H)$, which is given by the following constraints:
\begin{align*}
	(x_u(k,x),u(k))\in &\mathbb{Z}, & k&\in\mathbb{I}_{[0,N-1]}\\
	\sum_{i=j}^{T-1} H_i + \sum_{k=0}^{j-1} h(x_u(k,x),u(k))\leq &0, &  j&\in\mathbb{I}_{[1,T-1]}\\
	\sum_{k=i}^{i+T-1} h(x_u(k,x),u(k))\leq &0, &  i&\in\mathbb{I}_{[0,N-T]}.
\end{align*}
This yields the following MPC optimization problem 
\begin{equation}\label{eq:opt_problem}
	J^\ast_N(x,H):= \inf_{u\in\mathbb{U}^N(x,H)} J_N(x,u),
\end{equation} where $J^\ast_N(x,H)$ denotes the value function. We assume that the infimum is attained by an unique minimizer $u^\ast_{N,x,H}$. 
In closed-loop operation, the optimization problem~\eqref{eq:opt_problem} is solved in each time step $k$ and the first element of the optimal input is applied creating an implicit feedback law $\mu: \mathbb{X}\times\mathbb{H}\to\mathbb{U}$ satisfying $\mu_N(x,H):=u^\ast_{N,x,H}(0)$. The corresponding closed-loop system is given by $x(k+1)=f(x(k),\mu_N(x(k),H(k)))$ and $H(k+1)=\left[H_2(k), \ \dots, \ H_{T-1}(k), \ h(x(k), \mu_N(x(k),H(k)))  \right]$. We abbreviate for the closed loop $x_{\mu_N}(k,x,H)$, $H^{\mathrm{cl}}(k,x,H)$, $\mu_N(k,x,H)$ and $h_{\mu_N}(k,x,H)$. The cost of the closed loop over some time $K$ is given by
\begin{equation}\label{eq:def_CL_cost_functional}
	J^{\mathrm{cl}}_K(x,H):= \sum_{k=0}^{K-1} \ell(x_{\mu_N}(k,x,H),\mu_N (k,x,H)).
\end{equation}

\subsubsection*{Definitions:}The satisfaction of the transient average constraints implies that it holds with $k_{T,N}:=\left\lceil \frac{N}{T} \right\rceil T -N$
\begin{equation}\label{eq:bound_h_H}
\sum_{k=0}^{N-1} h(x_u(k,x),u(k))\leq - \sum_{i=1}^{k_{T,N}} H_{T-i},
\end{equation} for any $u\in\mathbb{U}^{N}(x,H)$. Since we know that the transient average constraints need to be satisfied for multiples of $T$, we can bound arbitrary time intervals by the $T-1$ previous values which are stored in $H$. In order to compare different storages of the transient average constraints, we need a norm-like measure, because usual norms are not reasonable since they do not take the signs of the entries into account. For the case of $H$ it is vital to consider the sign of the entries; just entries $H_{i,j}> 0$ should contribute to our measurement. To this end, we denote $\hat{H}_k=\sum_{i=1}^p \max\{ H_{i,k},0\}$ and define the norm-replacement $\Hnorm:=\max_{k\in\mathbb{I}_{[1,T-1]}}\{ \hat{H}(k)\}$. 
This has the property that $\Hnorm\geq 0$ holds and $\Hnorm=0$ implies that all previous $T-1$ time steps satisfy $h(x,u)\leq0$. 

We consider the case where the system is optimally operated at the optimal steady-state given by $
	\ell(x_s,u_s):= \min\{ \ell(x,u) ~\vert~ (x,u)\in\mathbb{Z},~ h(x,u)\leq 0,~ x=f(x,u) \}.$
\begin{defi}
	System~\eqref{eq:system} is \textit{optimally operated at steady-state} $(x_s,u_s)$, if for each initial condition $(x,H)\in\mathbb{X}\times\mathbb{H}$ and any input $u\in\mathbb{U}^{\infty}(x,H)$ it holds
	\begin{equation*}
		\liminf_{\tau\rightarrow\infty}\frac{1}{\tau}\sum_{k=0}^{\tau-1}\ell(x_u(k,x),u(k))\geq \ell(x_s,u_s).
	\end{equation*}
\end{defi} A sufficient condition for optimal operation at steady-state is dissipativity~\citep{Mueller13,Angeli12}. 
The following dissipativity assumption is taken from~\cite{Mueller14}. 
\begin{ass}\label{ass:dissip}
System~\eqref{eq:system} is \textit{strictly dissipative} on $\mathbb{Z}$ with supply rate $s(x,u):=\ell(x,u)-\ell(x_s,u_s)+\overline{\lambda}^\top h(x,u)$, i.e., 
there exists a a bounded storage function $\lambda: \,\mathbb{X}\to\mathbb{R}$, a multiplier $\bl\in\mathbb{R}^p_{\geq 0}$ and a function $\rho\in\mathcal{K}_\infty $ s.\,t. for all $(x,u)\in\mathbb{Z}$ it holds $\lambda(f(x,u))-\lambda(x)\leq s(x,u)-\rho(\|(x-x_s,u-u_s)\|)$.
Moreover, $\lambda$ is continuous, i.\,e. there exists a function $\alpha_\lambda\in\mathcal{K}_\infty$ with $\left| \lambda(x_1)-\lambda(x_2)  \right|\leq \alpha_\lambda(\norm{x_1-x_2})$ and w.\,l.\,o.\,g. $\lambda(x_s)=0$.
\end{ass}
We denote $\ell_s:=\ell(x_s,u_s)$, $h_s:= h(x_s,u_s)$ and $H^s:= \left[h_s,\ \dots, \ h_s \right]\in\mathbb{H}$ if the past values of the auxiliary outputs were at the steady-state. 
\section{Turnpike properties}\label{sec:Turnpikes}
In this section, we extend the turnpike properties from \citep{Gruene13} to EMPC subject to transient average constraints. Since not only the initial state $x$ but moreover, the past $T-1$ time steps are of interest, we provide a turnpike property for consecutive time instants in Theorem~\ref{thm:turnpike_T_consec}. In order to prove our assertions, we define the rotated stage cost
\begin{equation}\label{eq:rotated_cost}
\tilde{\ell}(x,u):=\ell(x,u)-\ell_s+\lambda(x)-\lambda(f(x,u))+\bl^\top h(x,u)
\end{equation} and obtain the corresponding rotated cost functional $\tilde{J}_{N}(x,u):= J_N (x,u)-N\ell_s + \lambda(x)-\lambda(x_u(N,x))+ \sum_{k=0}^{N-1} \bl^\top h(x_u(k,x),u(k))$. 
Note that Ass.~\ref{ass:comp_cont} and~\ref{ass:dissip} imply that there exist functions $\alpha_h,\,\alpha_u\in\mathcal{K}_\infty$ such that it holds for all $(x,u)\in\mathbb{Z}$: 
\begin{subequations}\label{eq:propositions_bounds}
	\begin{align}
			\vert \bl^\top h(x,u) \vert&\leq \alpha_h \left( \norm{x-x_s}+\norm{u-u_s}\right),\label{eq:propositions_bounds_h}\\
			0\leq\tilde{\ell}(x,u)&\leq \alpha_u \left( \norm{x-x_s}+\norm{u-u_s}\right),\label{eq:propositions_bounds_tilde_l}
	\end{align}
\end{subequations} which follows from continuity and $\bl^\top h_s=0$, compare~\citep[Lem.~1]{thesis_mario}.

Before introducing the turnpike property in Lem.~\ref{lem:turnpike}, we define bounds on the auxiliary output in combination with the multiplier $\bl$ which read $\underline{\vartheta}_h:=\inf_{(x,u)\in\mathbb{Z}}\bl^\top h(x,u)$ and $\bar{\vartheta}_h:=\sup_{(x,u)\in\mathbb{Z}}\bl^\top h(x,u)$. Furthermore, we denote the set of time instants for which the trajectory $x_u$ resulting from the control sequence $u\in\mathbb{U}^N(x,H)$ is in a neighborhood $\mathcal{B}_{\epsilon}(x_s,u_s)$ of the steady-state by
\begin{equation*}
	\mathcal{P}^\epsilon (u,x):= \{k\in\mathbb{I}_{[0,N-1]}|~(x_u(k,x),u(k))\in\mathcal{B}_{\epsilon}(x_s,u_s)\},
\end{equation*} and the number of time instants by $Q^\epsilon (u,x):=\# P^\epsilon (u,x)$. The subsequent lemma shows the so called turnpike property, which follows from strict dissipativity.
\begin{lemma}\label{lem:turnpike}
Let Ass.~\ref{ass:comp_cont} and~\ref{ass:dissip} hold. For each $(x,H)\in\mathbb{X}\times\mathbb{H}$, each $\epsilon,\delta>0$ and each $u\in\mathbb{U}^N(x,H)$ satisfying $J_N(x,H)\leq N \ell_s+\delta$, it holds $Q^\epsilon(u,x)\geq N-\frac{C'}{\rho(\epsilon)}$ with $C':=\delta+C-k_{T,N}\underline{\vartheta}_h$, $C:=2\sup_{x\in\mathbb{X}}|\lambda(x)|$.
\end{lemma}
\begin{proof}
	It follows from~\eqref{eq:bound_h_H}--\eqref{eq:rotated_cost} and the presumed conditions that $\tilde{J}_N(x,u)\leq \delta+C-\bl^\top \sum_{i=1}^{k_{T,N}}H_{T-i}\leq C'$. Now, assume $Q^\epsilon(u,x)<N-C'/\rho(\epsilon)$, i.\,e., there exists a set $\mathcal{N}\subseteq\mathbb{I}_{[0,N-1]}$ of $\#\mathcal{N}=N-Q^\epsilon(u,x)$ time instants s.\,t. $(x_u(k,x),u(k))\notin\mathcal{B}_{\epsilon}(x_s,u_s)$ holds for all $k\in\mathcal{N}$. Strict dissipativity implies $\tilde{J}_N(x,u)\geq (N-Q^\epsilon(u,x))\rho(\epsilon)>C'$ which is a contradiction and hence, yields the assertion.
\end{proof}
In  case $T=1$ we recover the results in~\cite[Thm. 5.3]{Gruene13}. Furthermore, in case  $h_i(x_s,u_s)<0$ for all $i\in\mathbb{I}_{[1,p]}$ we get $\bl=0$ and the following proofs are analogous to~\cite{Gruene13}. However, we want to consider the general case where the transient average constraints are active at the steady-state.

Another condition we need is an asymptotic controllability assumption w.r.t. the stage costs $\ell,\tilde{\ell}$ similar to~\cite[Ass.~5.5]{Gruene13}.
\begin{ass}\label{ass:asy_contr}
	There exist $\beta_1,\, \beta_2\in\mathcal{KLS}$ such that for each $(x,H)\in\mathbb{X}\times \mathbb{H}$ and each $N\in\mathbb{N}$ there exists  $u\in\mathbb{U}^{N}(x,H)$ such that it holds for all $k\in\mathbb{I}_{[0,N-1]}$:
{\begin{align*}
\ell(x_u(k,x),u(k))- \ell_s\leq&\beta_1 (\norm{x-x_s},k)+\beta_2(\HnormB{H-H^s},k)\\
\tilde{\ell}(x_u(k,x),u(k))\leq& \beta_1 (\norm{x-x_s},k)+\beta_2(\HnormB{H-H^s},k).
\end{align*}}
\end{ass}
 Note that Ass.~\ref{ass:asy_contr} yields that for optimal input sequences the condition $J_N^\ast (x,H)\leq N\ell_s+\delta$ always holds with $\delta:=\max_{x\in\mathbb{X}} \gamma_{\beta_1}(\norm{x-x_s})+\max_{H\in\mathbb{H}}\gamma_{\beta_2}(\HnormB{H-H^s})$ and we write $\hat{C}':= \delta+C-(T-1)\underline{\vartheta}_h$. 
Now, considering $q\in\mathbb{N}$ (different) trajectories at once, we introduce a set which contains all common time instants for which the trajectories are in a neighborhood $\mathcal{B}_\epsilon(x_s,u_s)$. Given $q$ trajectories $u_i\in\mathbb{U}^{N_i}(x_i,H^i)$, we write for the intersection $\mathcal{P}^{'\epsilon}_{[k_l,k_u]}\left( (u_1,x_1),\dots,(u_q,x_q) \right):= \cap_{i=1}^q\mathcal{P}^{\epsilon} (u_i,x_i) \cap \{ k_l,\dots,k_u \}$, where $0\leq k_l<k_u\leq  \max_i N_i=:\hat{N}$ are given bounds to focus on a specific interval. Using this and repeatedly considering optimal trajectories, we are able to show in the following theorem that there exists a sufficiently large prediction horizon such that the optimal trajectory has $T$ consecutive time instants in an arbitrarily small neighborhood $\mathcal{B}_\epsilon (x_s,u_s)$.

\begin{theo}\label{thm:turnpike_T_consec}
	Let Assumption~\ref{ass:comp_cont}-\ref{ass:asy_contr} hold. For any trajectories $u_i\in\mathbb{U}^{N_i}(x_i,H^i)$ with $i\in\mathbb{I}_{[1,q]}$, $(x_i,H^i)\in\mathbb{X}\times\mathbb{H}$ and $N_i\in\mathbb{N}$ satisfying $J_{N_i}(x_i,u_i)\leq N_i\ell_s+\delta$, as well as  for any $k_l,\,k_u\in\mathbb{N}$ s.\,t. $0\leq k_l<k_u\leq \hat{N}-1$ and any $m\in\mathbb{N}$ satisfying $k_u-k_l-m-q\Delta_N>0$ with $\Delta_N:=\max_{i,j\in\mathbb{I}_{[1,q]}} (N_i -N_j)$, the neighborhood
	\begin{equation}\label{eq:intersection_nbhd}
	\epsilon=\rho^{-1} \left( \frac{q\hat{C}'}{k_u-k_l-m-q\Delta_N} \right)
	\end{equation} yields $\# \mathcal{P}^{'\epsilon}_{[k_l,k_u]}\left( (u_1,x_1),\dots,(u_q,x_q) \right)\geq m$.  
	Furthermore, there exists $\sigma_T\in\mathcal{L}_{\mathbb{N}}$, such that for any $(x,H)\in\mathbb{X}\times\mathbb{H}$ and any $k'_l,\,k'_u\in\mathbb{N}$ satisfying $k'_l\in\mathbb{I}_{[0,k'_u-T^2-1]}$,  $k'_u\in\mathbb{I}_{[k'_l+T^2+1,N-1]}$, there exists $k_x\in\mathbb{I}_{[k'_l+T-1,k'_u]}$ such that the optimal trajectory $u^\ast_{N,x,H}\in\mathbb{U}^N(x,H)$ satisfies
	\begin{subequations}\label{eq:turnpike_assertion1}
		\begin{align}
			\norm{(x_{u^\ast_{N,x,H}}(k,x)-x_s,\, u^\ast_{N,x,H}(k)-u_s)}&\leq \epsilon \label{eq:turnpike_assertion1_x,u},\\
			\norm{h(x_{u^\ast_{N,x,H}}(k,x),u^\ast_{N,x,H}(k))-h_s}&\leq L_h \epsilon,  \label{eq:turnpike_assertion1_h}
		\end{align}
	\end{subequations} for all $k\in\mathbb{I}_{[k_x-T+1,k_x]}$ where $\epsilon:=\sigma_T (k'_u-k'_l-T^2)$. 
\end{theo}
\begin{proof}
	\textit{Part I: Set Intersection.}
	We make use of the turnpike property from Lem.~\ref{lem:turnpike} and get $Q^{\epsilon}(u_i,x_i)\geq N_i-\frac{C'_i}{\rho(\epsilon)}$ with $C'_i:=\delta+C-k_{T,N_i}\underline{\vartheta}_h\leq \hat{C}'$. Now, considering the neighborhood~\eqref{eq:intersection_nbhd} yields $Q^{\epsilon}(u_i,x_i)\geq\hat{N}-\frac{1}{q}(k_u-k_l-m) $ for all $i\in\mathbb{I}_{[1,q]}$ which guarantees by combinatorially using set intersections that the intersection contains at least $m$ elements.\\
	\textit{Part II: Showing assertion~\eqref{eq:turnpike_assertion1}.} Given the optimal trajectory $u^\ast_{N,x,H}=:u^\ast_0$, we consider  $T-1$ shifted trajectories  $u_i^\ast\in\mathbb{U}^{N_i}(x_i,H^i)$ for $i\in\mathbb{I}_{[1,T-1]}$ with $N_i=N-i$, $x_i=x_{u^\ast_0}(i,x)$ and $H^i = [ H_{i+1},\dots,H_{T-1},\\h(x_{u^\ast_0}(0,x),u^\ast_0(0)),\dots,\, h(x_{u^\ast_0}(i-1,x),u^\ast_0(i-1)) ]$ for $i\in\mathbb{I}_{[1,T-1]}$ where $H_j$ denotes the $j$-th column of $H$ and $H^0=H$. Since end pieces of optimal trajectories are again optimal, the trajectories $u_i$ are the optimal trajectories for initial condition $(x_i,H^i)$ and horizon $N_i$, i.\,e., we obtain
	\begin{equation}\label{eq:proof_T_cons}
		x_{u^\ast_i} (k,x_i)=x_{u^\ast_0}(k+i,x), \qquad u^\ast_i(k)=u^\ast_0 (k+i)
	\end{equation} for all $k\in\mathbb{I}_{[0,N_i-1]}$, $i\in\mathbb{I}_{[1,T-1]}$. Now, Ass.~\ref{ass:asy_contr} ensures that $J^\ast_{N_i}(x_i,H^i)\leq N \ell_s+\delta$ holds for any $i$. Hence, by considering the trajectories $x_{u^\ast_i}(\cdot,x_i)$ we can use the first part of the theorem with $q=T$, $\Delta_N=T-1$, $m=1$ and note $k_l=k'_l$ as well as $k_u=k'_u-(T-1)$. Choosing the neighborhood~\eqref{eq:intersection_nbhd} with our previous choices ensures that the intersection $\mathcal{P}^{'\epsilon}_{[k'_l,k'_u-(T-1)]}((u^\ast_0,x_0),\dots,(u^\ast_{T-1},x_{T-1}))$ contains at least one element which we denote by $k_x-T+1$. Now,~\eqref{eq:proof_T_cons} implies~\eqref{eq:turnpike_assertion1_x,u} with $\epsilon=\rho^{-1} \left(\frac{T\hat{C}'}{k'_u-k'_l-T^2}\right)=:\sigma_T(k'_u-k'_l-T^2)$.	Furthermore,~\eqref{eq:turnpike_assertion1_h} immediately follows from~\eqref{eq:turnpike_assertion1_x,u} using Lipschitz continuity of $h(x,u)$ (Ass.~\ref{ass:comp_cont}).	
\end{proof}
Note that we can choose $k'_u$ large enough (for a sufficiently large $N$) such that the assertion holds for any $\epsilon>0$ since $\sigma_T\in\mathcal{L}_{\mathbb{N}}$. 
Furthermore, this $T$-step consecutive turnpike in~\eqref{eq:turnpike_assertion1_h} also implies a bound on the extended state $H(k_x)$. In particular, we can upper bound the norm-like measure by using~\eqref{eq:turnpike_assertion1_h}. By definition we obtain
\begin{equation*}
	\Hnorm \leq \HnormB{H-H^s} \leq \norm{H-H^s}_1\leq \sqrt{p}L_h \epsilon
\end{equation*} and hence, satisfaction of $\HnormB{H-H^s} \leq E$ for any $E>0$ can be guaranteed by ensuring that~\eqref{eq:turnpike_assertion1} holds with $\epsilon\leq \frac{E}{\sqrt{p}L_h}$.
\section{Local continuity value function}\label{sec:loc_cont}
In this section, similar to~\citep{Gruene13}, we use a local controllability property to provide local continuity bounds on the value function in Thm.~\ref{thm:loc_cont}.
\begin{ass}\label{ass:loc_contr}
	There exist $\delta_c,\,E_h>0$, $d\in\mathbb{I}_{\geq T}$ and $\gamma_x,\,\gamma_u\in\mathcal{K}_\infty$ such that for each trajectory $x_{u_c} (k,x_c)$ with $u_c\in\mathbb{U}^{d+T}(x_c,H^c)$ satisfying $\HnormB{H^c}\leq E_h$ and $x_{u_c}(k,x_c)\in\mathcal{B}_{\delta_c}(x_s)$ for all $k\in\mathbb{I}_{[0,d+T]}$, the following holds:\\
	 For any trajectory $u_1\in\mathbb{U}^{N_1}(x_1,H^1)$ with $N_1\in\mathbb{I}_{\geq d+T}$, $(x_1,H^1)\in\mathbb{X}\times\mathbb{H}$, $x_3:=x_{u_1}(d,x_1)$ and $H^3:= H(x_{u_1}(d,x_1), u_1(d:d+T-2))$ satisfying $\HnormB{H^3-H^c(d+T-1)}\leq E_h$ and $x_3 \in\mathcal{B}_{\delta_c}(x_s)$, and for any $x_2\in\mathcal{B}_{\delta_c}(x_s)$ and any $H^2$ satisfying. $\HnormB{H^2-H^c}\leq E_h$, there exists an input $u_2\in\mathbb{U}^d(x_2,H^2)$ with $x_{u_2}(d,x_2)=x_3$ such that $u_3$ with $u_3(k)=u_2(k)$ for $k\in\mathbb{I}_{[0,d-1]}$ and $u_3(k)=u_1(k)$ for $k\in\mathbb{I}_{[d,N_1-1]}$ satisfies $u_3\in\mathbb{U}^{N_1}(x_2,H^2)$ and moreover, it holds for all $k\in\mathbb{I}_{[0,d]}$:
\begin{align*}
&\norm{x_{u_2}(k,x_2)-x_{u_c}(k,x_c)}\leq \gamma_x(\zeta),~
\norm{u_2(k)-u_c(k)}\leq \gamma_u (\zeta),
\end{align*}
\begin{align*}
\zeta:=\max \big\{& \norm{x_2-x_c}+\HnormB{H^2-H^c}, \,\\
& \norm{x_3-x_{u_c}(d,x_c)}+\HnormB{H^3-H^c(d+T-1)} \big\}.
\end{align*}
\end{ass}
Note that the Assumptions~\ref{ass:comp_cont}, \ref{ass:dissip} and~\ref{ass:loc_contr} ensure the existence of $\gamma_c,\,\gamma_h\in\mathcal{K}_\infty$ such that it holds for all $k\in\mathbb{I}_{[0,d]}$ and for $x_c$, $x_2$, $u_c(\cdot)$, $u_2(\cdot)$ and $\zeta$ from Ass.~\ref{ass:loc_contr}
\begin{equation}\label{eq:propositions_LCA}
\begin{split}
\left| \ell(x_{u_2}(k,x_2),u_2(k))-\ell(x_{u_c}(k,x_c),u_c(k)) \right| & \leq \gamma_c (\zeta),\\
\left| \tilde{\ell}(x_{u_2}(k,x_2),u_2(k))-\tilde{\ell}(x_{u_c}(k,x_c),u_c(k)) \right| & \leq \gamma_c (\zeta),\\
\left| \bl^\top\left[h(x_{u_2}(k,x_2),u_2(k))-h(x_{u_c}(k,x_c),u_c(k)) \right]\right| & \leq \gamma_h (\zeta),
\end{split}
\end{equation} which follows from continuity~\citep[Prop.~4]{thesis_mario}. Similar to~\cite[Ass.~6.2]{Gruene13}, Assumption~\ref{ass:loc_contr} ensures that given two states $(x_2,H^2)$, $(x_3,H^3)$ close to the optimal steady-state $(x_s,H^s)$, there exists an input trajectory $u_1$, such that we can drive the system from $x_2$ to $x_3$ and then apply any feasible input $u_3\in \mathbb{U}^{N_3}(x_3,H^3)$, while respecting the transient average constraints~\eqref{eq:trans_av_con}.

Now, analogous to~\cite{Gruene13} we can formulate a (turnpike) result for initial conditions in a steady-state neighborhood which implies that some consecutive points of the optimal trajectory stay close to the steady-state.
\begin{lemma}\label{lem:turnpike_P(N)}
	Suppose that Ass.~\ref{ass:comp_cont}-\ref{ass:loc_contr} hold. There exist $N_\eta\in\mathbb{N}$, a function  $\eta: \mathbb{N}\times \mathbb{R}_{\geq0} \to \mathbb{R}_{\geq0}$ with $\eta (N,r)\to 0$ if $N\to \infty$ and $r\to 0$ such that for any $x \in \mathcal{B}_{\delta_c} (x_s)$, any $H\in\mathbb{H}$ with $\HnormB{H-H^s}\leq E_h$ and horizon $N\geq N_\eta$, the optimal trajectory $u^\ast_{N,x,H}\in\mathbb{U}^N(x,H)$ satisfies for all $k\in\mathbb{I}_{[0,\frac{N}{2}+T-1]}$
	\begin{align*}
	&\norm{(x_{u^\ast_{N,x,H}} (k,x)-x_s, \,u^\ast_{N,x,H}(k)-u_s )} \\\leq &\eta \left( N, \norm{x-x_s}+ \HnormB{H-H^s} \right).
	\end{align*}
\end{lemma}
\begin{proof}
	The proof is similar to~\cite[Lem.~6.3]{Gruene13}. By using local controllability, we construct a candidate sequence which steers the system to the steady-state $(x_s,H^s)$ and after $k_y-d$ steps, with $k_y\in\mathbb{I}_{[\frac{N}{2},N-T]}$, from the steady-state back to the optimal trajectory. Feasibility of this candidate can be ensured by a sufficient large prediction horizon and thus $\epsilon$ small (cf. Thm.~\ref{thm:turnpike_T_consec}). Finally, the assertion follows from dissipativity and a proof of contradiction. Details can be found in \cite[Lem.~2]{thesis_mario}.
\end{proof}
Using Lemma~\ref{lem:turnpike_P(N)}, we can show that $J^\ast_N$ is locally continuous for sufficiently large prediction horizons $N$. 
\begin{theo}\label{thm:loc_cont} Let Ass.~\ref{ass:comp_cont}-\ref{ass:loc_contr} hold. There exist $N_2\in\mathbb{N}$ and $\gamma_v\in\mathcal{K}_\infty$ such that for all $\delta\in(0,\delta_c]$, all $E\in(0,E_h]$, all $N\geq N_2$, all $x\in\mathcal{B}_\delta (x_s)$ and all $\norm{H-H^s}_1\leq E$ it holds
	\begin{equation}\label{eq:turnpike_assertion2}
	\left| J^\ast_N(x,H)-J^\ast_N(x_s,H^s) \right| \leq \gamma_v(\delta+E).
	\end{equation}
\end{theo}
\begin{proof}
	Using Lemma~\ref{lem:turnpike_P(N)}, we choose  $N_2\geq N_\eta$  such that $\eta(N,r)\leq\min\{\delta_c, \frac{E_h}{\sqrt{p}L_h}\}$ holds for all $N\geq N_2$ and $r\in[0,\delta_c+E_h]$ and we abbreviate $u^\ast:=u^\ast_{N,x,H}$. With this choice of $N$, Lemma~\ref{lem:turnpike_P(N)} ensures $x_{u^\ast}\in\mathcal{B}_{\delta_c}(x_s)$ for $k\in\mathbb{I}_{[0,d+T]}$ as well as  $\HnormB{H(d+T-1)-H^s}\leq E_h$, $\HnormB{H^s-H}\leq E_h$. Hence, we can apply Ass.~\ref{ass:loc_contr}. First, we connect the steady-state trajectory with $x_{u^\ast}(d,x)$ which is possible due to Ass.~\ref{ass:loc_contr}, i.\,e., there exists $u_2\in\mathbb{U}^d(x_s,H^s)$ such that $x_{u_2}(d,x_s)=x_{u^\ast}(d,x)$ and $u_2(k)=u^\ast(k),~k\in\mathbb{I}_{[d,N-1]}$. It follows from~\eqref{eq:propositions_LCA} that $\ell(x_{u_2}(k,x_s),u_2(k))\leq \ell(x_{u^\ast}(k,x),u^\ast(k))+\gamma_c (\norm{x-x_s}+\HnormB{H^s-H})$ holds for all $k\in\mathbb{I}_{[0,d-1]}$. This yields $J^\ast_N(x_s,H^s)\leq J_N(x_s,u_2)\leq J^\ast_N(x,H)+d\gamma_c (\norm{x-x_s}+\norm{H-H^s}_1)$. \\
	The upper bound on $J^\ast_N(x,H)$ can be constructed similarly using a trajectory connecting $(x,H)$ with the optimal trajectory starting at the steady-state $(x_s,H^s)$, which proves~\eqref{eq:turnpike_assertion2} with $\gamma_v:=d\gamma_c$.
\end{proof}

\section{Performance guarantees}\label{sec:peformance}
In this section, we derive performance guarantees in terms of value convergence of the closed-loop cost $J^{\mathrm{cl}}_K$ from~\eqref{eq:def_CL_cost_functional} and its rotated analogue $\tilde{ J}^{\mathrm{cl}}_K$. Similar to the set $\mathcal{P}^\epsilon$, we define the set $\mathcal{T}^\epsilon_{[a,b]}(u,x):= \{ k_x\in\mathbb{I}_{[a,b]}: \, (x_u (k_x-i,x), u(k_x-i))\in\mathcal{B}_\epsilon(x_s,u_s), \,\forall i \in\mathbb{I}_{[0,T-1]} \}$ and the intersection set  $\mathcal{T}^{'\epsilon}_{[a,b]}((u_1,x_1),\dots,(u_q,x_q)):= \cap_{i=1}^q\mathcal{T}^{\epsilon}_{[a,b]} (u_i,x_i)$. We introduce the following assumption which holds, e.g., for exponentially stabilizable systems with $\tilde{\ell}$ quadratic.
\begin{ass}\label{ass:fct_psi}
	There exists $\overline{\psi}\in\mathcal{K}_\infty$ such that for any $k_\psi\in\mathbb{I}_{\geq T-1}$, any $N\in\mathbb{N}$, all $(x,H)\in\mathbb{X}\times\mathbb{H}$  it holds:\\
	$-\bl^\top \sum_{k=0}^{k_\psi} h(x_{\tilde{ u}_{N,x,H}^\ast}(k,x), \tilde{ u}_{N,x,H}^\ast(k)) \leq \psi (x,H),$ where we abbreviate $\psi(x,H):=\overline{\psi}(\norm{x-x_s}_1+\norm{H-H^s}_1)$.
\end{ass}
We remark that this assumption is similar to the asymptotic controllability property (Ass.~\ref{ass:asy_contr}), but this time for the auxiliary output $h$ and for the optimal trajectory $\tilde{u}^\ast$.
The following lemma bounds the difference in the open-loop cost of the original problem and the rotated problem  over the first $k_x$ steps. 
\begin{lemma}\label{lem:residuum}
	Let Ass.~\ref{ass:comp_cont}-\ref{ass:fct_psi} hold. There exist  $N_5\in\mathbb{N}$, $R_1,\,R_2\in\mathcal{K}_\infty$ and $\tilde{\sigma}_T\in\mathcal{L}_{\mathbb{N}}$ such that for all $(x,H)\in\mathbb{X}\times\mathbb{H}$, any $N\geq N_5$ with  $\epsilon=\tilde{\sigma}_T(N-N_2-T(2T-1))$ there exist  $k_x\in\mathcal{T}^{'\epsilon}_{[T-1,N-N_2]}((u^\ast_{N,x,H},x),(\tilde{u}^\ast_{N,x,H},x))$. Furthermore, it holds for $J_{k_x}:=J_{k_x}(x,u^\ast_{N,x,H})$,  $\tilde{J}_{k_x}:=\tilde{J}_{k_x}(x,\tilde{u}^\ast_{N,x,H})+k_x\ell(x_s,u_s)-\lambda(x)$: 
	\begin{subequations}\label{eq:ass_residuum}
		\begin{align}
		&\left| J_N^\ast(x,H)-J_{k_x} -J^\ast_{N-k_x}(x_s,H^s) \right|\leq R_1(\epsilon), \label{eq:residuum_1}\\
		& J_{k_x}-\tilde{J}_{k_x}\leq \psi(x,H)+R_2(\epsilon),\label{eq:residuum_2a}\\
		&\tilde{J}_{k_x}-J_{k_x}\leq (T-1)\norm{\bl}L_h\epsilon+  R_2(\epsilon).\label{eq:residuum_2b}
		\end{align} 
	\end{subequations}
\end{lemma}
\begin{proof}
The proof is an extension of~\cite[Sec.~7]{Gruene13}. The existence of $k_x$ follows using turnpike properties for consecutive time instants for a horizon $N\geq N_5:=N_2+\frac{2T\hat{C}'}{\rho(\bar{\epsilon})}+T(2T-1)$  with $ \bar{\epsilon}:=\min\{\delta_c, E_h/(\sqrt{p}L_h)\}$ and
 $\tilde{\sigma}_T:=\rho^{-1} \left(\frac{2T\hat{C'}}{N-N_2-T(2T-1)}\right)$. The bound~\eqref{eq:residuum_1} then follows using local continuity of the value function (Thm.~\ref{thm:loc_cont}), for details see~\cite[Ch.~5]{thesis_mario}. In order to show~\eqref{eq:residuum_2a},~\eqref{eq:residuum_2b} we need to bound the auxiliary output $\overline{\lambda}^\top h$. 
 By using Ass.~\ref{ass:fct_psi}, we obtain a lower bound on the auxiliary output of the optimal trajectory of the rotated problem which shows~\eqref{eq:residuum_2a}.
 Furthermore, using~\eqref{eq:bound_h_H} in combination with $k_x\in\mathcal{T}^{'\epsilon}_{[T-1,N-N_2]}$ implies an upper bound which yields~\eqref{eq:residuum_2b}.
\end{proof}
In order to construct a local candidate sequence, we show that for the steady-state neighborhood from Thm.~\ref{thm:loc_cont}, there exists a single control step implying $\mathcal{K}_\infty$-bounds w.\,r.\,t. the initial condition $(x,H)$.
\begin{propo}\label{prop:bounds_f_l_h} Let Ass.~\ref{ass:comp_cont} and~\ref{ass:loc_contr} hold. There exist $ \gamma_f$, $\gamma_{l}$, $\gamma_y\in\mathcal{K}_\infty$ such that for   all $\delta\in(0,\delta_c]$, all $x\in\mathcal{B}_\delta (x_s)$, all $E\in(0,E_h]$ and all $\HnormB{H-H^s}\leq E$  there exists $u_x\in\mathbb{U}(x,H)$ such that $f(x,u_x)\in\mathbb{X}$ and it holds:
	\begin{subequations}\label{eq:ass_bounds_f_l_h}
		\begin{align}
		\norm{f(x,u_x)-x_s}&\leq \gamma_f (\delta+E),\label{eq:ass_bound_f_l_h_a)}\\
		\norm{h(x,u_x)-h_s}&\leq \gamma_y (\delta+E),\label{eq:ass_bound_f_l_h_b)} \\
		\ell(x,u_x)-\ell_s&\leq  \gamma_{l} (\delta+E).\label{eq:ass_bounds_f_l_h_cost}
		\end{align}
	\end{subequations}
\end{propo}
\begin{proof}
The assertion follows from Ass.~4 for $(u_c,x_c,H^c)=(u_s,x_s,H^s)$, $u_x=u_2(0)$ and $\gamma_f:=\gamma_x$, $\gamma_y:=L_h(\gamma_x+\gamma_u)$, $\gamma_l:=\alpha_l (\gamma_x+\gamma_u)$, using continuity from Ass.~\ref{ass:comp_cont}.
\end{proof}
The following theorem uses Prop.~1 and Thm.~2 to construct a feasible candidate solution and provide an upper bound on the closed-loop cost and closed-loop rotated cost. Considering the closed loop, we write for the extended state $\chi:=(x,H)$ and $\chi_{\mu_N}(k,\chi):=(x_{\mu_N}(k,\chi),H^{\mathrm{cl}}(k,\chi))$.
\begin{theo}\label{thm:performance}
	Let Ass.~\ref{ass:comp_cont}-\ref{ass:fct_psi} hold. There exist $N_6\in\mathbb{N}$ and $\sigma_3, \sigma_6\in\mathcal{L}_{\mathbb{N}}$ such that for all $(x,H)\in\mathbb{X}\times \mathbb{H}$, all $K\in\mathbb{N}$ and all $N\geq N_6+1$ it holds with $N_{6,1}:=N_2+T^2+1$ and $N_{6,2}:=N_2+(4T+1)(T-1)+1$
	\begin{subequations}
		\begin{align}
		\label{eq:value_conv_ass}
		J^{\mathrm{cl}}_K (\chi) \leq &J^\ast_N (\chi)-J^\ast_N(\chi_{\mu_N}(K,\chi))\\&+K(\ell_s+\sigma_3(N-N_{6,1})),\nonumber\\		
		\label{eq:value_conv_ass_rotated}
		\tilde{ J}_K^{\mathrm{cl}}(\chi)\leq &\tilde{ J}^\ast_N(\chi)-\tilde{ J}^\ast_N(\chi_{\mu_N}(K,\chi))+\sigma_6(N-N_{6,2})\\&+K\sigma_3 (N-N_{6,1})+\psi(x,H)+\sum_{k=0}^{K-1} \bl^\top h_{\mu_N}(k,\chi).\nonumber
		\end{align}
	\end{subequations}
\end{theo}
\begin{proof}
		
	\textit{Part I: Showing~\eqref{eq:value_conv_ass}}.
	We set	$k'_l=0$ and $k'_u=N-N_2\geq T^2+1$ and get from Thm.~\ref{thm:turnpike_T_consec} that there are $T$ consecutive time instants in the interval $\mathbb{I}_{[0,N-N_2]}$ in a nbhd. of the steady-state. Furthermore, we propose a candidate sequence  $\hat{u}_{N,\chi}=\hat{u}\in\mathbb{U}^{N+1}(\chi)$ which reads as follows: $\hat{u}(k)=u^\ast_{N,\chi}(k)$ for $k\in\mathbb{I}_{[0,k_x-1]}$, $\hat{u}(k_x)=u'_{x'}$ with $u'_{x'}$ as given in Prop.~\ref{prop:bounds_f_l_h} and $\hat{u}(k)=u^\ast_{N-k_x,\chi''}(k-k_x-1)$ for all $k\in\mathbb{I}_{[k_x+1,N]}$, where we abbreviate $x':=x_{u^\ast}(k_x,x)$, $x'':=f(x',u'_{x'})$ and $H'$ as well as $H''$ analogously, which yields $\chi':=(x',H')$, $\chi'':=(x'',H'')$. This candidate sequence is feasible for prediction horizons $N\geq N_2+\frac{T \hat{C}'}{\rho(\epsilon')}+T^2=:N_3$ with $\epsilon':= \min \{ \bar{\epsilon},\frac{1}{2}\gamma' ,\frac{\gamma'}{2\sqrt{p}L_h}\}$ where $\gamma':= \min \{ \gamma_f^{-1} (\delta_c), \gamma_y^{-1} (\frac{E_h}{\sqrt{p}}) \}$ and it holds $x',\,x''\in\mathcal{B}_{\delta_c}(x_s)$ as well as $\HnormB{H'-H^s},\, \HnormB{H''-H^s}\leq E_h$. Furthermore, we define $J'_N(x,\hat{u}):=\sum\limits_{\substack{k=0\\k\neq k_x}}^N \ell(x_{\hat{u}}(k,x),\hat{u}(k))$. Now, we apply Thm.~\ref{thm:loc_cont} for the time instants $k_x$ and $k_x+1$ such that it holds $\sum_{k=k_x+1}^N \ell(x_{\hat{u}}(k,x),\hat{u}(k))=J^\ast_{K_1}(\chi'')\leq J^\ast_{K_1}(\chi')+\delta_1 (N_d)$ with $\delta_1\in\mathcal{L}_{\mathbb{N}}$ where $K_1:=N-k_x\geq N_2$ and $N_d:=N-N_2-T^2$. In particular, we have
\begin{align*}
\delta_1:=\gamma_v (\gamma_f ((1+\sqrt{p}L_h)\sigma_T)+\sigma''_H)+\gamma_v((1+\sqrt{p}L_h)\sigma_T),
\end{align*}
 with $\sigma''_H :=\max\{ \sqrt{p}L_h \sigma_T , \sqrt{p}\gamma_y ((1+\sqrt{p}L_h)\sigma_T) \}$. We get
	\begin{equation}\label{eq:ineq_proof_value_convergence}
	\begin{split}
	J'_N(x,\hat{u})\leq &J^\ast_N(\chi)+\delta_1(N_d), \\
	\ell(x',u'_{x'})\leq &\ell_s+\delta_2(N_d),
	\end{split}
	\end{equation} where the second inequality follows from Prop.~\ref{prop:bounds_f_l_h} with $\delta_2:=\gamma_l ((1+\sqrt{p}L_h)\sigma_T )$.
	Considering the closed loop, we get from the dynamic programming principle~\citep{Bertsekas95}
	\begin{equation*}
		\begin{split}
		J_K^{\mathrm{cl}}(\chi)=&J^\ast_N (\chi)-J^\ast_{N-1}(\chi_{\mu_N}(K,\chi))\\&+\sum_{k=1}^{K-1} J^\ast_N(\chi_{\mu_N}(k,\chi))-J^\ast_{N-1}(\chi_{\mu_N}(k,\chi)).
		\end{split}
	\end{equation*}  
	By combining this with~\eqref{eq:ineq_proof_value_convergence} and the definition of $J'_N(x,\hat{u})$ we obtain
	\begin{align*}
	J^\ast_N(\chi_{\mu_N}(k,\chi))-J^\ast_{N-1}(\chi_{\mu_N}(k,\chi))\leq \ell_s+\sum_{i=1}^2 \delta_i (N-N_{6,1}),
	\end{align*}
	  which yields the assertions with $\sigma_3:=\delta_1+\delta_2$.\\
	\textit{Part II: Showing~\eqref{eq:value_conv_ass_rotated}.} 
	In order to apply~\eqref{eq:residuum_1} to the original and rotated problem at initial values $\chi$ and $\chi_{\mu_N}(K,\chi)$ we need at least one time instant
	\begin{equation*}
	\begin{split}
	k_x\in\mathcal{T}^{'\epsilon}_{[T-1,N-N_2]}((u^\ast_{N,\chi},x),(u^\ast_{N,\chi_{\mu_N}(K,\chi)},x_{\mu_N}(K,\chi)),\\(\tilde{u}^\ast_{N,\chi},x),(\tilde{u}^\ast_{N,\chi_{\mu_N}(K,\chi)},x_{\mu_N}(K,\chi)))
	\end{split}
	\end{equation*} in a neighborhood $\epsilon\leq \bar{\epsilon}$. We choose $N_6:=\max \{ N_{6,2}+\frac{4T\hat{C}'}{\rho(\bar{\epsilon})},~N_3+1\}$, which implies $\epsilon=\rho^{-1} (\frac{4T\hat{C}'}{N-N_{6,2}})=:\sigma_{6,1}(N-N_{6,2})$ with $\sigma_{6,1}\in\mathcal{L}_{\mathbb{N}}$. Now, we can apply Lemma~\ref{lem:residuum} and use the results of Part I of this proof which yields
	\begin{align}
	J^\ast_N(\chi)&-J^\ast_N(\chi_{\mu_N}(K,\chi))\leq\tilde{ J}^\ast_N(\chi)-\tilde{ J}^\ast_N(\chi_{\mu_N}(K,\chi))\nonumber\\&-\lambda(x)+\lambda(x_{\mu_N}(K,\chi))+R(\epsilon)+\psi(x,H)\label{eq:proof_value_rotated_cost}
	\end{align}
	with $ R(\epsilon)= 8 \gamma_v((1+\sqrt{p}L_h)\epsilon)+2\alpha_\lambda (\epsilon)+(T-1)L_h \norm{\bl}\epsilon$, $ R\in\mathcal{K}_\infty$. With $\epsilon=\sigma_{6,1}(N-N_{6,2})$, we obtain $\sigma_6(N-N_{6,2}):=R(\sigma_{6,1}(N-N_{6,2}))$ with $\sigma_6\in\mathcal{L}_{\mathbb{N}}$. Now, using the definition of the rotated cost, \eqref{eq:value_conv_ass} and~\eqref{eq:proof_value_rotated_cost} yields~\eqref{eq:value_conv_ass_rotated}. Additional details can be found in~\cite[Thm.~5 \& 6]{thesis_mario}.
\end{proof}
\begin{remark}
Using the value convergence results from Theorem~\ref{thm:performance} for the rotated cost~\eqref{eq:value_conv_ass_rotated}, one can directly establish practical convergence of the closed loop using $\tilde{\ell}$ positive definite, compare~\cite[Thm.~7 \& 8]{thesis_mario}.
\end{remark}
\section{Stability results}\label{sec:Stability}
In this section we conclude practical asymptotic stability (p.\,a.\,s.) of the proposed EMPC scheme. Additional to the optimal rotated value function, which is sufficient to show p.\,a.\,s. for EMPC without transient average constraints~\citep{Gruene13}, we use input-to-state stability (ISS) of the state $H$ which results in a practical non-monotonic Lyapunov function. Finally, the approach from~\cite{Ahmadi2008} is used in order to obtain a practical Lyapunov function.
\subsubsection*{Input-to-State-Stability of the Storage}\ \\
From one time step to another, the storage $H$ is shifting its columns one to the left and the updated last column is equal to the auxiliary output $h(x,u)$ at the last time instant. Therefore, we obtain a discrete-time system of the form $H(k+1)=f_H(H(k),x(k),u(k))$. 
\begin{lemma}\label{lem:ISS}
	For any $\kappa\in(0,\infty)$, the function $\hat{V}_\kappa(H):=\sum_{i=1}^{T-1} i \norm{H_i-h_s}_1^\kappa$ satisfies 
\begin{align*}
\norm{H-H^s}_1^\kappa \leq& \hat{V}_\kappa(H)\leq (T-1)^2 \norm{H-H^s}_1^\kappa,\\
&\hat{V}_\kappa(f_H(H,x,u))-\hat{V}_\kappa(H)\\
\leq& -\norm{H-H^s}_1^\kappa+(T-1)\norm{h(x,u)-h_s}_1^\kappa.
\end{align*}
 Furthermore, $H$ is ISS w.\,r.\,t. $h(x,u)-h_s$.
\end{lemma}
\begin{proof}
	It holds 
\begin{align*}
\norm{H-H^s}_1^\kappa= \max_{j\in\mathbb{I}_{[1,T-1]}}\{\norm{H_j-h_s}_1^\kappa\}\leq \hat{V}_\kappa(H),
\end{align*}
 as well as 
 \begin{align*}
\hat{V}_\kappa (H)\leq (T-1) \sum_{i=1}^{T-1} \norm{H_i-h_s}_1^\kappa\leq (T-1)^2 \norm{H-H^s}_1^\kappa.
\end{align*}
 Furthermore, we can bound 
 \begin{align*}
&\hat{V}_\kappa(f_H(H,x,u))-\hat{V}_\kappa (H)\\
=&-\sum_{j=1}^{T-1} \norm{H_j-h_s}_1^\kappa+(T-1)\norm{h(x,u)-h_s}_1^\kappa \\
\leq& -\norm{H-H^s}^\kappa_1 + (T-1) \norm{h(x,u)-h_s}_1^\kappa
\end{align*}
  and ISS follows from~\cite[Lem. 3.5]{Jiang2001}.
\end{proof} 
\begin{remark}\label{rem:bound_future_h}
	Since the transient average constraints also need to be satisfied in the overlapping periods, i.\,e., also w.\,r.\,t. the past values of $h(x,u)$, we can upper bound 
	\begin{equation*}
		\bl^\top \sum_{j=1}^{T-1} \sum_{k=0}^{j-1} h_{\mu_N} (k,x,H) \leq (T-1)^2 \norm{\bl}\cdot \norm{H-H^s}_1. 
	\end{equation*}
\end{remark}

\subsubsection*{Practical Asymptotic Stability}\ \\
We consider p.\,a.\,s. of the extended state $(x,H)$ as defined in~\cite{Gruene14}.
\begin{defi}
	The steady-state $(x_s,H^s)$ is called \textit{practically asymptotically stable} w.\,r.\,t. $\epsilon\geq 0$ on a set $\mathcal{S}\subseteq \mathbb{X}\times \mathbb{H}$ with $(x_s,H^s)\in\mathcal{S}$ if there exitsts $\beta\in\mathcal{KL}$ such that $\norm{x_{\mu_N}(k,x,H)-x_s}+\norm{H^{\mathrm{cl}}(k,x,H)-H^s}\leq \max\{ \beta(\norm{x-x_s}+\norm{H-H^s},k) , \ \epsilon\} $ holds for all $(x,H)\in\mathcal{S}$ and all $k\in\mathbb{N}_0$.
\end{defi}
In order to construct a practical Lyapunov function, we require a polynomial lower bound on $\rho\in\mathcal{K}_\infty$ (from Ass.~\ref{ass:dissip}). Furthermore, we assume that the function $\psi\in\mathcal{K}_\infty$ (from Ass.~\ref{ass:fct_psi}) satisfies a suitable upper bound.
\begin{ass}\label{ass:polynomial}
	There exist constants $a,\omega>0 $ such that $\rho(r)\geq a\cdot r^\omega$ holds for all $r\in[0,r_{\mathrm{max}}]$ with $r_{\mathrm{max}}:=\max_{(x,u)\in\mathbb{Z}}\norm{(x-x_s,u-u_s)}$. Furthermore, we have
\begin{align*}
\psi(x,H)< \frac{1}{2}a(n+m)^{-\frac{\omega}{2}}\left( \norm{x-x_s}_1^\omega +\frac{\norm{H-H^s}_1^\omega}{L_h(T-1)}\right)
\end{align*}
	for $T\geq 2$ and all $(x,H)\in\mathbb{X}\times\mathbb{H}$ with $(x,H)\neq (x_s,H^s)$.
\end{ass}
We point out that the previous assumption is only needed in order to show practical asymptotic stability; performance and convergence guarantees (Sec.~\ref{sec:peformance}) have been shown without this assumption. Finally, the following theorem provides a practical Lyapunov function. 
\begin{theo}\label{thm:p.a.s}
	Let Ass.~\ref{ass:comp_cont}-\ref{ass:polynomial} hold and $T\geq 2$. There exist $c>0$, functions $\alpha_1,\alpha_2,\alpha_3\in\mathcal{K}_\infty$ and $\delta_1,\delta_2\in\mathcal{L}_{\mathbb{N}}$ s.\,t. \begin{equation*}
	\begin{split}
		W(x,H):= &\sum_{j=0}^{T-1} \hat{W}(x_{\mu_N}(j,x,H),H^{\mathrm{cl}}(j,x,H)),\\
		\hat{W}(x,H):= &\tilde{ J}^\ast_N(x,H)+c\sum_{i=1}^{T-1} i \cdot \norm{H_i-h_s}_1^\omega
	\end{split}
	\end{equation*} satisfies with $\xi:=\norm{x-x_s}+\norm{H-H^s}_1$
	\begin{align}
		&\alpha_1(\xi)\leq W(x,H)\leq \alpha_2(\xi)+\delta_1(N-N_6),\label{eq:bounds_W}\\
		\begin{split}
		&W(x_{\mu_N}(1,x,H),H^{\mathrm{cl}}(1,x,H))\\\leq &W(x,H)-\alpha_3(\xi)+\delta_2(N-N_6),\label{eq:decrease_W}
		\end{split}
	\end{align}
	for all $(x,H)\in\mathbb{X}\times\mathbb{H}$ and all $N\geq N_6+1$. Moreover, the steady-state $(x_s,H^s)$ is p.\,a.\,s. for all $(x,H)\in\mathbb{X}\times\mathbb{H}$ w.\,r.\,t. $\epsilon\to0$ as $N\to\infty$.
\end{theo}
\begin{proof}
	We split this proof in three different parts. In the first part we investigate the rotated value function. Then, we combine $\tilde{J}^\ast_N(x,H)$ with the ISS property of $H$ in order to obtain a non-monotonic practical Lyapunov function $\hat{W}(x,H)$. In the third, part we construct $W(x,H)$ by using $\hat{W}(x,H)$ and show that it satisfies the bounds~	\eqref{eq:bounds_W}--\eqref{eq:decrease_W}, which implies p.\,a.\,s.\\
	\textit{Part I: Optimal Rotated Value Function.} Strict dissipativity implies $\tilde{ J}^\ast_N(x,H)\geq \rho(\norm{(x-x_s,u-u_s)})$. Using a case distinction of $(x,H)$, we can construct a candidate sequence in order to get an upper bound. 
We use either asymptotic controllability (Ass.~\ref{ass:asy_contr}) or local controllability (Ass.~\ref{ass:loc_contr})  to obtain a feasible candidate $u_1$ that drives the system to the optimal steady-state, which by optimality implies $\tilde{J}^\ast_N(x,H)\leq \tilde{ J}_N(x,u_1)\leq \max\{ \frac{\delta}{\min \{\delta_c,E_h\} }\xi, \ d\alpha_u (\gamma_x(\xi)+\gamma_u(\xi)) \}=:\alpha_7(\xi)$ with $\alpha_7\in\mathcal{K}_\infty$. 
Combining the two bounds, we obtain
	\begin{equation}\label{eq:bounds_J_tilde}
		\rho(\norm{(x-x_s,u-u_s)})\leq \tilde{ J}^\ast_N(x,H)\leq \alpha_7(\xi).
	\end{equation} By using dissipativity and Thm.~\ref{thm:performance} with $K=1$, we obtain that it holds with $N\geq N_6+1$ and
 	for all $(x,H)\in\mathbb{X}\times\mathbb{H}$	\begin{align}\nonumber
	\tilde{ J}^\ast_N(\chi_{\mu_N}(1,\chi))\leq \tilde{ J}^\ast_N(\chi)+\delta_7(N-N_6)+\psi(x,H)\\-\rho(\norm{(x-x_s,\mu_N(x,H)-u_s)})+\bl^\top h(x,\mu_N(x,H)),\label{eq:decrease_J_tilde}
	\end{align} with $\delta_7:=\sigma_6+\sigma_3$. For more details we refer to~\cite[Lem.~7]{thesis_mario}.\\
	\textit{Part II: Combination of $\tilde{ J}^\ast_N(x,H)$ and ISS.} 
	We set
	\begin{equation*}
		\hat{W}(x,H):=\tilde{ J}^\ast_N(x,H)+c\hat{V}_\omega (H),
	\end{equation*}
	 with $c:=\frac{a(n+m)^{-0.5\omega}}{2L_h(T-1)}>0$. By using comparison function properties~\citep{Kellett2014}, we obtain that there exist $\hat{\alpha}_1(\xi)\leq \rho(\norm{x-x_s})+c\norm{H-H^s}_1^\omega$ and $\hat{\alpha}_2(\xi)\geq \alpha_7(\xi)+c(T-1)^2 \norm{H-H^s}_1^\omega$, which yields from~\eqref{eq:bounds_J_tilde} and Lem.~\ref{lem:ISS}
	\begin{equation}\label{eq:bounds_hat_W}
		\hat{\alpha}_1 (\xi)\leq \hat{W}(x,H)\leq \hat{\alpha}_2(\xi).
	\end{equation} Now, it follows from~\eqref{eq:decrease_J_tilde}, ISS, Lipschitz continuity of $h(x,u)$ and the definition of $\hat{W}(x,H)$ that it holds 
\begin{align*}
&\hat{W}(\chi_{\mu_N}(1,\chi))-\hat{W}(\chi)+\rho(\norm{(x-x_s,\mu_N(x,H)-u_s)}) \\
\leq&\delta_7(N-N_6)+\bl^\top h_{\mu_N}(0,x,H)-c\norm{H-H^s}_1^\omega\\
&+c(T-1)L_h \norm{(x-x_s,\mu_N(x,H)-u_s)}_1^\omega+\psi(x,H)
\end{align*}
	 By using Ass.~\ref{ass:polynomial}, we obtain with our choice of $c$ that there exists $\hat{\alpha}_3\in\mathcal{K}_\infty$ such that it holds for all $N\geq N_6+1$
	\begin{equation}
	\begin{split}\label{eq:decrease_hat_W}
		&\hat{W}(\chi_{\mu_N}(1,\chi))-\hat{W}(\chi) \\\leq &-\hat{\alpha}_3 (\xi) + \delta_7(N-N_6)+ \bl^\top h(x,\mu_N(x,H)).
	\end{split}
	\end{equation}
	Note that the last term summed over $T$ steps is always negative, which implies that $\hat{W}(x,H)$ is a practical Lyapunov function over $T$ steps (non-monotonical).  \\
	\textit{Part III: Practical Lyapunov Function $W(x,H)$.} We construct a monotonically decreasing function based on~\citep{Ahmadi2008}. In particular, we use $W(\chi):=\sum_{j=0}^{T-1}$ $\hat{W}(\chi_{\mu_N}(j,\chi))$ which is equal to $W(x,H)$ given in the assertion. Using $\hat{W}(x,H)\geq 0$ from~\eqref{eq:bounds_hat_W} yields the lower bound of~\eqref{eq:bounds_W} with $\alpha_1:=\hat{\alpha}_1$.
	From~\eqref{eq:decrease_hat_W} and Rem.~\ref{rem:bound_future_h}, we obtain 
\begin{align*}
W(x,H)\leq& T\hat{\alpha}_2 (\xi)+(T-1)^2 \norm{\bl}\cdot \norm{H-H^s}_1\\
&+\frac{T(T-1)}{2}\delta_7 (N-N_6),
\end{align*}
 which shows~\eqref{eq:bounds_W} with $\alpha_2(\xi):=T\hat{\alpha}_2 (\xi)+(T-1)^2\norm{\bl} \xi$ and $\delta_1:=\frac{T}{2}(T-1)\delta_7$.
	 Furthermore, using the definition of the transient average constraints and~\eqref{eq:decrease_hat_W} yields $W(\chi_{\mu_N}(1,\chi))-W(\chi)\leq -\hat{\alpha}_3(\xi)+T\delta_7(N-N_6)$ which shows~\eqref{eq:decrease_W} with $\alpha_3:=\hat{\alpha}_3$ and $\delta_2:=T\delta_7$. Now, practical asymptotic stability directly follows from~\cite[Prop. 4.3]{Faulwasser18} with respect to $\epsilon(N-N_6)$, $\epsilon:=\alpha_1^{-1}(\alpha_2(\alpha_3^{-1}\delta_2+\delta_2)+\delta_1+\delta_2)\in\mathcal{L}_{\mathbb{N}}$ and hence, $\epsilon\to0$ as $N\to\infty$.
\end{proof}
As previously mentioned,~\eqref{eq:decrease_J_tilde} shows that the rotated value function is not a valid Lyapunov function for the EMPC setting subject to transient average constraints~\eqref{eq:trans_av_con}. However, the function $W(x,H)$ is a valid Lyapunov function for the extended state $(x,H)$. We conjecture that stability of transient average constrained EMPC with terminal conditions~\citep{Mueller14} can be shown using similar arguments.
\section{Numerical Example}\label{sec:Example}
In this section, we illustrate some of the provided theoretical results. We consider the example from~\cite{Mueller14b,Koehler17} which reads $x(k+1)=x(k)u(k)$ with state and input constraint set $\mathbb{Z}:=\mathbb{X}\times\mathbb{U}:=[-10,10]^2$ and transient average constraints of the form~\eqref{eq:trans_av_con} with $y=h(x,u)=2x+u-5$. The stage cost reads $\ell(x,u)=(x-3)^2+u^2$ which implies that the optimal steady-state is given by $(x_s,u_s)=(2,1)$. Thus, Ass.~\ref{ass:comp_cont} holds with the Lipschitz constant $L_h=3$. 
Strict dissipativity holds with $\bl=1$, the continuous storage function $\lambda(x)=1.5(x-2)$ and  $\rho(r)=0.25r^2\in\mathcal{K}_\infty$ which immediately satisfies the conditions in Ass.~\ref{ass:polynomial} with $a=0.25$ and $\omega=2$. 
The asymptotic controllability condition is difficult to show as stated, but we conjecture that the results can be modified such that asymptotic controllability on a control invariant sublevel set is sufficient using arguments from~\cite{Boccia14,JoKo18}. For a consideration of the local controllability property (Ass.~\ref{ass:loc_contr}) at the given example we refer to~\cite[Ch.~7]{thesis_mario}.

\subsubsection*{Turnpike Properties}\ \\
We consider an initial value $x=1$ where we also stay the past $T-1$ values there, i.\,e., $H=[h(1,1), \ \dots, \ h(1,1)]\in\mathbb{R}^{1\times(T-1)}$. Fixing $T=3$, we investigate two different prediction horizons $N_1=10$ and $N_2=12$. Simulations show that the amount of time instants in a steady-state nieghborhood is increasing for a larger $N$ as it is stated in Lem.~\ref{lem:turnpike}. Moreover, the neighborhood of the steady-state is shrinking for an increasing $N$ as it is shown in~\cite[Thm. 2]{thesis_mario}. Now, we fix $N=12$ and vary the time periods of the transient average constraints $T_1=3$ and $T_2=6$. The numerical result shows that the trajectory resulting from a larger time period $T$ is allowed to stay longer in a ``cheap'' region w.\,r.\,t. the stage cost $\ell$ which implies that the steady-state neighborhood is increasing for a larger $T$ which is in accordance with the results of Thm.~\ref{thm:turnpike_T_consec}. 

\subsubsection*{Closed Loop Results}\ \\
Now, we investigate the given EMPC scheme for $N=12$, $T=6$, $x=2$ and $H=[h(1,1), \ \dots, \ h(1,1), \ h(1,2)]\in\mathbb{R}^{1\times5}$. As shown in Theorem~\ref{thm:p.a.s}, the closed loop converges to a neighborhood of the optimal steady-state. Moreover, the rotated value function is not decreasing over the MPC iterations; but as proved in Thm.~\ref{thm:p.a.s}, the novel Lyapunov function $W(x,h)$ is (practically) monotonically decreasing. This result is illustrated in Figure~\ref{fig:ex_W_J_tilde_cl}.
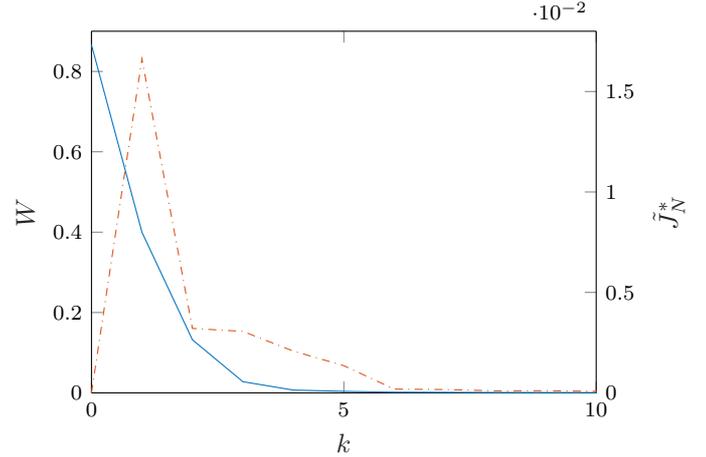
\begin{figure}
	\definecolor{mycolor1}{rgb}{0.00000,0.44700,0.74100}%
	\definecolor{mycolor2}{rgb}{0.85000,0.32500,0.09800}%
	\centering
%
%
\tikzstyle{every node}=[font=\small]
\definecolor{mycolor1}{rgb}{0.00000,0.44700,0.74100}%
\definecolor{mycolor2}{rgb}{0.85000,0.32500,0.09800}%
\def\plotwidth{0.75\linewidth}
\def\xlim{10}
\begin{tikzpicture}[%
trim axis left, trim axis right
]
\pgfplotsset{set layers}

\begin{axis}[%
width=\plotwidth,
height=1.888in,
scale only axis,
xmin=0,
xmax=\xlim,
xlabel style={font=\color{white!15!black}},
xlabel={$k$},
xtick = {0,5,10,15,20},
ymin=0,
ymax=0.018,
axis y line*=right,
ylabel style = {align=center},
ylabel style={font=\color{white!15!black}},
ylabel={$\tilde{J}^\ast_N$},
axis background/.style={fill=white}
]
\addplot [color=mycolor2, forget plot, dashdotted]
  table[row sep=crcr]{%
0	1.0392909643997e-08\\
1	0.0166221538547582\\
2	0.00320679473529495\\
3	0.00306059980881912\\
4	0.00207951698865827\\
5	0.00135492016678995\\
6	0.000186437078632196\\
7	0.000166248740505637\\
8	0.000101689191115639\\
9	0.000108349072352354\\
10	7.84095541135343e-05\\
11	1.30485525779989e-05\\
12	2.63955872696897e-05\\
14	6.8645055222305e-06\\
15	1.42653455803554e-05\\
16	2.42603006839204e-06\\
19	4.65798219551061e-07\\
25	7.95871898162659e-07\\
29	3.90251660320473e-07\\
};\label{plot:ex_J_tilde}
\end{axis}
\begin{axis}[%
width=\plotwidth,
height=1.888in,
scale only axis,
xmin=0,
xmax=\xlim,
xlabel style={font=\color{white!15!black}},
xlabel={$k$},
ymin=0,
ymax=0.9,
axis y line*=left,
axis x line=none,
ylabel style = {align=center},
ylabel style={font=\color{white!15!black}},
ylabel={$W$},
axis background/.style={fill=white}
]
\addplot [color=mycolor1, forget plot]
table[row sep=crcr]{%
	0	0.866310666607585\\
	1	0.399947701463773\\
	2	0.1320694866622\\
	3	0.0278293150107878\\
	4	0.00698929105750778\\
	5	0.00378743728139241\\
	6	0.00164478637747578\\
	8	0.000620627588620692\\
	12	0.000113625443987075\\
	23	4.13786687047946e-06\\
};\label{plot:ex_W}
\end{axis}
\end{tikzpicture}%
 \caption{Lyapunov function $W$~\eqref{plot:ex_W} and rotated value function $\tilde{J}^\ast_N$~\eqref{plot:ex_J_tilde} along the closed loop.}
	\label{fig:ex_W_J_tilde_cl}
\end{figure}
\section{Conclusion}\label{sec:Conclusion}
In this work, we investigated transient average constrained EMPC without terminal constraints and showed performance guarantees as well as practical asymptotic stability. First, we introduced an additional state storing past values of the auxiliary output in order to consider the transient average constraints. We provided a turnpike phenomenon for consecutive time instants and by using a local controllability property, local continuity of the value function as well as convergence of the closed-loop cost (original cost and rotated cost) was shown. As the main contribution, we proved  practical asymptotic stability by a combination of the rotated value function, ISS of the auxiliary output storage and using results on non-monotonic Lyapunov functions from~\cite{Ahmadi2008}.

\bibliography{Literature}  

\end{document}